\documentclass{llncs}
\usepackage{amssymb,amsfonts}
\usepackage{booktabs,tabularx}
\usepackage{url,xspace,soul,wrapfig}
\usepackage{graphicx,hyperref,paralist}
\usepackage[font=small, hypcap=true]{subfig,caption}
\usepackage{paralist,enumerate}
\usepackage[textsize=footnotesize]{todonotes}
\usepackage{enumerate}
\usepackage{soul,lineno}

\let\doendproof\endproof
\renewcommand\endproof{~\hfill$\qed$\doendproof}

\makeatother

\graphicspath{{figures/}}

\title{3D Visibility Representations\\ of $1$-planar Graphs}

\author{Patrizio Angelini\inst{1} \and Michael A. Bekos\inst{1} \and \\Michael Kaufmann\inst{1} \and Fabrizio Montecchiani\inst{2}}

\institute{
Institut f{\"u}r Informatik, Universit{\"a}t T{\"u}bingen, Germany
\\\email{\{angelini,bekos,mk\}@informatik.uni-tuebingen.de}
\and
Dipartimento di Ingegneria, Universit\'a degli Studi di Perugia, Italy
\\\email{fabrizio.montecchiani@unipg.it}
}

\begin{document}
\maketitle

\pagenumbering{arabic}

\begin{abstract} 
We prove that every $1$-planar graph $G$ has a $z$-parallel visibility representation, i.e., a 3D visibility representation in which the vertices are isothetic disjoint rectangles parallel to the $xy$-plane, and the edges are unobstructed $z$-parallel visibilities between pairs of rectangles. In addition, the constructed representation is such that there is a plane that intersects all the rectangles, and this intersection defines a bar $1$-visibility representation of $G$.
\end{abstract}

\section{Introduction}
\label{sec:introduction}

Visibility representations are a classic research topic in Graph Drawing and Computational Geometry. Motivated by VLSI applications, seminal papers studied \emph{bar visibility representations}  of planar graphs (see, e.g.,~\cite{DBLP:journals/dcg/RosenstiehlT86,DBLP:journals/dcg/TamassiaT86,t-prg-84,DBLP:conf/compgeom/Wismath85}), in which vertices are represented as  non-overlapping horizontal segments, called \emph{bars}, and edges correspond to vertical \emph{visibilities} connecting pairs of bars, i.e., vertical segments that do not intersect any bar other than at their endpoints.

In order to represent non-planar graphs, more recent papers investigated models in which either two visibilities are allowed to cross, or a visibility can ``go through'' a vertex. Two notable examples are rectangle visibility representations and bar $k$-visibility representations. In a \emph{rectangle visibility representation} of a graph, every vertex is represented as an axis-aligned rectangle and two vertices are connected by an edge using either a horizontal or a vertical visibility (see, e.g.,~\cite{DBLP:journals/dam/DeanH97,DBLP:journals/comgeo/HutchinsonSV99,DBLP:conf/cccg/Shermer96}). A \emph{bar $k$-visibility representation} is a bar visibility representation in which each visibility intersects at most $k$ bars (see, e.g.,~\cite{DBLP:journals/jgaa/Brandenburg14,DBLP:journals/jgaa/DeanEGLST07,DBLP:journals/jgaa/Evans0LMW14}).

Extensions of visibility representations to 3D have also been studied. Of particular interest for us are \emph{$z$-parallel visibility representations (ZPRs)}, in which the vertices of the graph are isothetic disjoint rectangles parallel to the $xy$-plane, and the edges are visibilities parallel to the $z$-axis. Bose et al.~\cite{DBLP:journals/jgaa/BoseEFHLMRRSWZ98} proved that $K_{22}$ admits a ZPR, while $K_{56}$ does not. {\v{S}}tola~\cite{DBLP:conf/gd/Stola08} reduced this gap by showing that $K_{51}$ does not admit any ZPR. If the rectangles are restricted to unit squares, then $K_7$ is the largest representable complete graph~\cite{Fekete1995}. Other 3D visibility models are  box visibility representations~\cite{DBLP:journals/ijcga/FeketeM99}, and  2.5D box visibility representations~\cite{DBLP:conf/gd/ArleoBGEGLMMWW16}.

\begin{figure}[t]
    \centering
    \subfloat[\label{fig:g}{}]
    {\includegraphics[width=0.33\columnwidth,page=1]{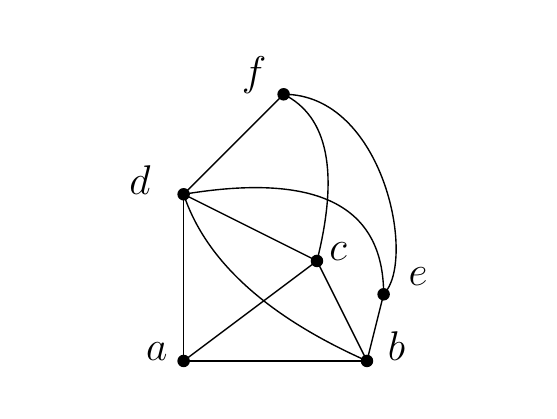}}
    \hfil
    \subfloat[\label{fig:sgamma}{}]
    {\includegraphics[width=0.33\columnwidth,page=2]{figures/example}}
    \hfil
    \subfloat[\label{fig:gamma}{}]
    {\includegraphics[width=0.33\columnwidth,page=3]{figures/example}}

    \caption{%
    (a) A $1$-planar graph $G$.
    (b) The intersection of a $1$-visible ZPR $\gamma$ of $G$ with the plane $Y=0$; the red (bold) visibilities traverse a bar.
    (c) The projection to the $yz$-plane of $\gamma$ (only the red visibilities are shown).}
    \label{fig:example}
\end{figure}

In this paper we study 3D visibility representations of $1$-planar graphs. We recall that a graph is \emph{$1$-planar} if it can be drawn with at most one crossing per edge (see, e.g.,~\cite{bsw-bsgr-AMS83,DBLP:journals/corr/KobourovLM17,pt-gdfce-C97}). The  $1$-planar graphs are among the most investigated families of ``beyond planar graphs'', i.e., graphs that extend planarity by forbidding specific edge crossings configurations (see, e.g.,~\cite{hong_et_al:DR:2017:7038,DBLP:conf/ictcs/Liotta14}). Brandeburg~\cite{DBLP:journals/jgaa/Brandenburg14} and Evans et al.~\cite{DBLP:journals/jgaa/Evans0LMW14} proved that every $1$-planar graph admits a bar $1$-visibility representation. Later, 
Biedl et al.~\cite{DBLP:conf/compgeom/BiedlLM16}  proved that a $1$-plane graph (i.e., an embedded $1$-planar graph) admits a rectangle visibility representation if and only if it does not contain any of a set of obstructions, and that not all $1$-planar graphs can be realized,  regardless of their $1$-planar embedding. On the other hand, every $1$-planar graph can be represented with vertices that are orthogonal polygons with several reflex corners~\cite{DBLP:conf/gd/GiacomoDELMMW16}. 
Our goal is to represent $1$-planar graphs with vertices drawn as rectangles (rather than more complex polygons) by exploiting the third dimension. We prove that every $1$-planar graph $G$ has a ZPR $\gamma$.  In addition, $\gamma$ is \emph{$1$-visible}, i.e., there is a plane that is orthogonal to the rectangles of $\gamma$ and such that its intersection with $\gamma$ defines a bar $1$-visibility representation of $G$ (see Section~\ref{sec:preliminaries} for formal definitions). 

Our main contribution is summarized by the following theorem.
\begin{theorem}
Every $1$-planar graph $G$ with $n$ vertices admits a $1$-visible ZPR $\gamma$ in $O(n^3)$ volume. Also, if a $1$-planar embedding of $G$ is given as part of the input, then $\gamma$ can be computed in $O(n)$ time.
\label{thm:main}
\end{theorem}

An embedding is needed, as recognizing $1$-planar graphs is \textsc{NP}-complete~\cite{DBLP:journals/algorithmica/GrigorievB07,DBLP:journals/jgt/KorzhikM13}. An example of a $1$-visible ZPR is shown in Fig.~\ref{fig:example}. We also remark that, as pointed out by Kobourov et al. in a recent survey~\cite{DBLP:journals/corr/KobourovLM17}, very little is known on 3D representations of $1$-planar graphs, and our result sheds some light on this problem.

From a high-level perspective,  to prove Theorem~\ref{thm:main} (see Section~\ref{sec:main}) we start by constructing a bar $1$-visibility representation $\gamma_1$ of $G$, which is then used as the intersection of the ZPR $\gamma$ with the plane $Y=0$ (see, e.g., Fig.~\ref{fig:sgamma}). In particular, we transform each bar $b$ of $\gamma_1$ into a rectangle $R_b$ by computing the $y$-coordinates of its top and bottom sides, so that each visibility in $\gamma_1$ that traverses a bar $b$ can be represented as a visibility in $\gamma$ that passes above or below $R_b$ (see, e.g., Fig.~\ref{fig:gamma}). This is done by using two suitable acyclic orientations of the edges of~$G$. 

Some proofs and technicalities have been moved to the appendix.


\section{Preliminaries and definitions}
\label{sec:preliminaries}
We assume familiarity with the concepts of planar drawings and planar embeddings, see, e.g.,~\cite{DBLP:books/ph/BattistaETT99}.
The \emph{planarization} of a non-planar drawing is a planar drawing obtained by replacing every crossing with a \emph{dummy vertex}.  An \emph{embedding} of a graph is an equivalence class of  drawings whose planarized versions have the same planar embedding. A \emph{$1$-plane} graph is a 1-planar graph with a \emph{$1$-planar embedding}, i.e., an embedding where each edge is incident to at most one dummy vertex. A \emph{kite} is a $1$-plane graph isomorphic to $K_4$ in which the outer face is composed of four vertices and four crossing-free edges, while the remaining two edges cross each other. Given a $1$-plane graph $G$ and a kite $K=\{a,b,c,d\}$, with $K \subseteq G$, kite $K$  is \emph{empty} if it  contains no vertex of $G$ inside the 4-cycle $\langle a,b,c,d \rangle$.

A \emph{(partial) orientation} $\mathcal O$ of a graph $G$ is an assignment of directions to (a subset of) the edges of $G$. The graph obtained by orienting the edges of $G$ according to  $\mathcal O$ is the directed (or mixed) graph $G_{\mathcal O}$. A \emph{planar $st$-(multi)graph} $G$ is a plane acyclic directed (multi)graph with a single source $s$ and a single sink $t$, with both $s$ and $t$ on its outer face~\cite{DBLP:journals/tcs/BattistaT88}. The sets of incoming and outgoing edges incident to each vertex $v$ of $G$ are \emph{bimodal}, i.e., they are contiguous in the cyclic ordering of the edges at $v$.  Each face $f$ of $G$  is bounded by two directed paths with a common origin and destination, called the \emph{left path} and \emph{right path} of $f$. Face $f$ is the \emph{left} (resp., \emph{right}) face for all vertices on its right (resp., left) path except for the origin and for the destination. A  \emph{topological ordering} of a directed acyclic (multi)graph is a linear ordering of its vertices such that for every directed edge from vertex $u$ to vertex $v$, $u$ precedes $v$ in the ordering.

A set $\mathcal R$ of disjoint rectangles in $\mathbb{R}^3$ is \emph{$z$-parallel}, if each rectangle has its sides parallel to the $x$- and $y$-axis. Two rectangles of $\mathcal R$ are \emph{visible} if and only if they contain the ends of a closed cylinder $C$ of radius $\varepsilon>0$ parallel to the $z$-axis and orthogonal to the $xy$-plane, and that does not intersect any other rectangle.

\begin{definition}
A \emph{$z$-parallel visibility representation (ZPR)} $\gamma$ of a graph $G$ maps the set of vertices of $G$ to a $z$-parallel set of disjoint rectangles, such that for each edge of $G$ the two corresponding rectangles are visible\footnote{Our visibility model is often called \emph{weak}, to be distinguished with the \emph{strong} model in which visibilities and edges are in bijection. While this distinction is irrelevant when studying complete graphs (e.g., in~\cite{DBLP:journals/jgaa/BoseEFHLMRRSWZ98,DBLP:conf/gd/Stola08}), the weak model is commonly adopted to represent sparse non-planar graphs in both 2D and 3D (see, e.g.,~\cite{DBLP:conf/gd/ArleoBGEGLMMWW16,DBLP:conf/compgeom/BiedlLM16,DBLP:journals/jgaa/Brandenburg14,DBLP:conf/gd/GiacomoDELMMW16,DBLP:journals/jgaa/Evans0LMW14}).}. If there is a plane that is orthogonal to the rectangles of $\gamma$ and such that its intersection with $\gamma$ defines a bar $k$-visibility representation of $G$, then  $\gamma$  is a \emph{$k$-visible ZPR}.
\end{definition}

\section{Proof of Theorem~\ref{thm:main}}
\label{sec:main}

Let $G=(V,E)$ be a $1$-plane graph with $n$ vertices. To prove Theorem~\ref{thm:main}, we present a linear-time algorithm that takes $G$ as input and computes a $1$-visible ZPR of $G$ in cubic volume. The algorithm works in three steps, described in the following. 

\medskip\noindent{\it Step 1.} We compute a bar $1$-visibility representation $\gamma_1$ of $G$ by applying Brandenburg's linear-time algorithm~\cite{DBLP:journals/jgaa/Brandenburg14}, which  produces a representation with integer coordinates on a grid of size $O(n^2)$. This algorithm consists of the following steps.
\begin{inparaenum}[a)]
\item A $1$-plane multigraph $G'=(V,E'\supseteq E)$ is computed from $G$ such that: The four end-vertices of each pair of crossing edges of $G'$ induce an empty kite; no edge can be added to $G'$ without introducing crossings; if two vertices are connected by a set of $k>1$ parallel edges, then all of them are uncrossed and non-homotopic. We remark that the embedding of $G'$ may differ from the one of $G$ due to the rerouting of some edges.
\item Let $P$ be the plane multigraph obtained from $G'$ by removing all pairs of crossing edges. Let $\mathcal O$ be an orientation of $P$ such that $P_{\mathcal O}$  is a planar $st$-multigraph. Then the algorithm by Tamassia and Tollis~\cite{DBLP:journals/dcg/TamassiaT86} is applied to compute a bar visibility representation of $P_{\mathcal O}$.
\item Finally, all pairs of crossing edges are reinserted through a postprocessing step that extends the length of some bars so to introduce new visibilities. The newly introduced visibilities traverse at most one bar each. In addition, each bar is traversed by at most one visibility.
\end{inparaenum}

\medskip\noindent{\it Step 2.} We transform each bar $b_v$ of $\gamma_1$ to a preliminary rectangle $R_v$. We assume that $\gamma_1$ lies on the $xz$-plane and that the bars are parallel to the $x$-axis. Let $z(v)$ be the $z$-coordinate of $b_v$ and let $x_L(v)$ and $x_R(v)$ be the $x$-coordinates of the left and right endpoints of $b_v$, respectively. The rectangle $R_v$ lies on the plane  parallel to the $xy$-plane with equation $Z=z(v)$. Also, its left and right sides have $x$-coordinates equal to $x_L(v)$ and $x_R(v)$, respectively. It remains to compute the $y$-coordinates of the top and bottom sides of $R_v$. We preliminarily set the $y$-coordinates of the bottom sides and of the top sides of all the rectangles to $-1$ and $+1$, respectively. All the visibilities of $\gamma_1$ that do not traverse any bar can be replaced with cylinders of radius $\varepsilon < \frac{1}{2}$. Let $P'$ be the subgraph of $G'$ induced by all such visibilities, and let  $\gamma_2$ be the resulting ZPR. The next lemma follows.

\begin{lemma}\label{le:gammaprime}
$\gamma_2$ is a ZPR of $P'$.
\end{lemma}
\medskip\noindent{\it Step 3.} To realize the remaining visibilities of $\gamma_1$, we  modify the $y$-coordinates of the rectangles. The idea is to define two partial orientations of the edges of $P$, denoted by $\mathcal O_1$ and $\mathcal O_2$,  to assign the final $y$-coordinates of the top sides and of the bottom sides of the rectangles, respectively. In particular, an edge oriented from $u$ to $v$ in $\mathcal O_1$ ($\mathcal O_2$) encodes that the top side (bottom side) of $R_u$ will have $y$-coordinate greater (smaller) than the one of $R_v$. The orientations are such that if two vertices $u$ and $v$ see each other through a third vertex $w$ in $\gamma_1$, then their top (bottom) sides both have larger (smaller) $y$-coordinate than the one of $w$. Hence, both $\mathcal O_1$ and $\mathcal O_2$ are defined based on $\gamma_1$, using the following three rules.

\begin{figure}[t]
    \centering
    \subfloat[\label{fig:right wing}{}]
    {\includegraphics[width=0.165\columnwidth,page=2]{figures/faces}}
    \hfil
    \subfloat[\label{fig:right wing-dr}{}]
    {\includegraphics[width=0.165\columnwidth,page=5]{figures/faces}}
    \hfill
    \subfloat[\label{fig:left wing}{}]
    {\includegraphics[width=0.165\columnwidth,page=1]{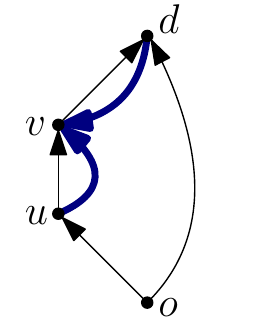}}
    \hfil
    \subfloat[\label{fig:left wing-dr}{}]
    {\includegraphics[width=0.165\columnwidth,page=4]{figures/faces}}
    \hfill
    \subfloat[\label{fig:diamond}{}]
    {\includegraphics[width=0.165\columnwidth,page=3]{figures/faces}}
    \hfil
    \subfloat[\label{fig:diamond-dr}{}]
    {\includegraphics[width=0.165\columnwidth,page=6]{figures/faces}}
    \caption{(a)-(b) A right wing. (c)-(d) A left wing. (e)-(f) A diamond. }
    \label{fig:faces}
\end{figure}

Let $f=\{o,u,v,d\}$ be a face of $P_{\mathcal O}$ (and hence of $P$) such that $\{o,u,v,d\}$ are part of an empty kite of $G'$.  In what follows we assume that $o$ is the origin and $d$ is the destination of the face. We borrow some terminology from~\cite{DBLP:journals/jgaa/Brandenburg14}, refer to Fig.~\ref{fig:faces} (the black thin edges only). If the left (resp., right) path of $f$ is composed of the single edge $(o,d)$, then $f$ is called a \emph{right wing} (resp., \emph{left wing}). If both the left path and the right path of $f$ consist of two edges, then $f$ is a \emph{diamond}.
\begin{inparaenum}[{\bf (R.}1{\bf )}]

\item\label{r:rightwing} If $f$ is a right wing, we may assume that $b_v$ is above $b_u$. Consider the restriction of $\gamma_1$ with respect to $\{o,u,v,d\}$. Either the visibility between $b_u$ and $b_d$ traverses $b_v$ (as in Fig.~\ref{fig:right wing-dr}), or  the visibility between $b_o$ and $b_v$ traverses $b_u$. In both cases we only orient edges in $\mathcal O_1$. In the first case we orient $(u,v)$ from $u$ to $v$ and $(v,d)$ from $d$ to $v$ (see the green bold edges in Fig.~\ref{fig:right wing}). In the second case we orient $(o,u)$ from $o$ to $u$ and $(u,v)$ from $v$ to $u$.
\item\label{r:leftwing} If $f$ is a left wing, we may assume that $b_v$ is above $b_u$. 
As for a right wing, either the visibility between $b_u$ and $b_d$ traverses $b_v$ (as in Fig.~\ref{fig:left wing-dr}), or  the visibility between $b_o$ and $b_v$ traverses $b_u$.  We orient the edges as for a right wing, but we only consider $\mathcal O_2$ (see, e.g., the blue bold edges in Fig.~\ref{fig:left wing}).
\item\label{r:diamond} If $f$ is a diamond, we may assume that $b_u$ is to the left of $b_v$. 
Either the visibility between $b_o$ and $b_d$ traverses $b_v$, or the visibility between $b_o$ and $b_d$ traverses $b_u$. In the first case we orient $(o,v)$ from $o$ to $v$ and $(v,d)$ from $d$ to $v$ in $\mathcal O_1$ (see the green bold edges in Fig.~\ref{fig:diamond}). In the second case we orient $(o,u)$ from $o$ to $u$ and $(u,d)$ from $d$ to $u$ in $\mathcal O_2$.
\end{inparaenum}

\noindent By applying the above three rules for all left and right wings, and for all diamonds of $P_{\mathcal O}$, we obtain $\mathcal O_1$ and $\mathcal O_2$. Note that the above procedure is correct, in the sense that no edge is assigned a direction twice. This is due to the fact that a direction in $\mathcal O_1$ (resp., $\mathcal O_2$) is assigned to an edge only if it belongs to the right (resp., left) path of a right (resp., left) wing  or of a diamond. On the other hand, an edge belongs only to one right path and to one left path.
In what follows, we prove that both $P_{\mathcal O_1}$ and $P_{\mathcal O_2}$ are acyclic, i.e., they have no oriented cycles.

\begin{lemma}\label{le:acyclic}
Both $P_{\mathcal O_1}$ and $P_{\mathcal O_2}$ are acyclic.
\end{lemma}
\begin{sketch}
We prove that  $P_{\mathcal O_1}$ is acyclic. The argument for  $P_{\mathcal O_2}$ is symmetric. Suppose, for a contradiction, that $P_{\mathcal O_1}$ contains a directed cycle $C =\langle e_1,e_2,\dots,e_c \rangle$, as shown in Fig.~\ref{fig:cycle}. First, note that $c>2$. If $c=2$, there are two non-homotopic parallel edges that are both part of the right path of a right wing or of a diamond in $P_\mathcal O$. But this is impossible since each pair of crossing edges in $G'$ forms an empty kite. 
\begin{figure}[t]
    \centering
    \subfloat[\label{fig:cycle}{$C$}]
    {\includegraphics[width=0.3\columnwidth,page=1]{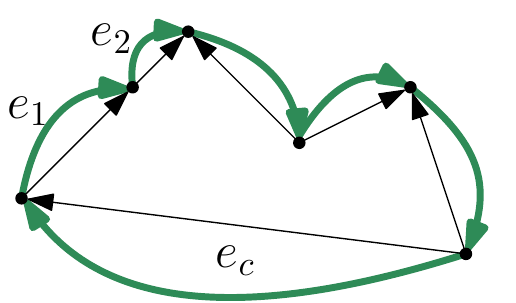}}
    \hfil
    \subfloat[\label{fig:cycle-c1}{{\bf Case 1.}}]
    {\includegraphics[width=0.3\columnwidth,page=2]{figures/cycle}}
    \hfil
    \subfloat[\label{fig:cycle-c2}{{\bf Case 2.}}]
    {\includegraphics[width=0.3\columnwidth,page=3]{figures/cycle}}
    \caption{Illustration for the proof of Lemma~\ref{le:acyclic}. In black (thin) we show the orientation of the edges according to $\mathcal O$, while in green (bold) according to $\mathcal O_1$.}
\end{figure}
 Some edges of $C$ have opposite orientations in $\mathcal O$ and $\mathcal O_1$, since $\mathcal O$ is acyclic. In particular, there is at least a non-empty maximal subsequence $S = \langle e_i,e_{i+1}, \dots, e_j \rangle$ of $C$ with this property. We distinguish two cases, whether  $C$ is oriented clockwise or counter-clockwise in a closed walk along its boundary. Let $a$ and $b$ be the origin of $e_i$ and the destination of $e_j$, respectively. Note that there is a directed path from $b$ to $a$ in $P_{\mathcal O}$ (and from $a$ to $b$ in $P_{\mathcal O_1}$).

\begin{inparaenum}[{\bf Case} 1.]

\item Refer to Fig.~\ref{fig:cycle-c1}. Since $e_j$ is oriented in $P_{\mathcal O_1}$, it belongs to the right path of a right wing or of a diamond $f$ of $P_{\mathcal O}$ by {\bf R.}\ref{r:rightwing} and {\bf R.}\ref{r:diamond}. Also, $b$ is the origin of $f$, as otherwise $b$ would have an incoming edge between $e_j$ and $e_{j+1}$  in counterclockwise order from $e_j$, which violates the bimodality of the edges around $b$ or the fact that the source $s$ of $P_{\mathcal O}$ is on the outer face. But then the orientation of $e_j$ in $\mathcal O_1$ contradicts {\bf R.}\ref{r:rightwing} or {\bf R.}\ref{r:diamond}. 

\item This case can be handled similarly by observing that $a$ is the destination of a face $f$ having $e_{i-1}$ in its right path.
\end{inparaenum}
\end{sketch}
For each maximal subsequence of the edges of $P_{\mathcal O_1}$ such that each edge is oriented and the induced subgraph is connected, compute a topological ordering. Concatenate all such topological orderings, and append at the beginning or at the end of the sequence possible vertices that are not incident to any oriented edge. This gives a total ordering  of the vertices of $P_{\mathcal O_1}$, denoted by $\sigma_1$. Set the $y$-coordinate of the top side of the rectangle representing the $i$-th vertex in $\sigma_1$ equal to $n-i+1$. Apply a symmetric procedure for $P_{\mathcal O_2}$, by computing a total ordering $\sigma_2$, and by setting the $y$-coordinate of the bottom side of the rectangle representing the $i$-th vertex in $\sigma_2$ equal to $i-n-1$. This concludes the construction of $\gamma$ (possible dummy edges inserted by the augmentation procedure of {\em Step 1(a)} are simply ignored in $\gamma$). The correctness of $\gamma$ easily follows.

\begin{lemma}\label{le:gamma}
$\gamma$ is a $1$-visible ZPR of $G$.
\end{lemma}
Since $\gamma_1$ takes $O(n^2)$ area, and each rectangle of $\gamma$ has height at most $2n$, it follows that $\gamma$ takes $O(n^3)$ volume. Also, each step of the algorithm can be performed in linear time. This concludes the proof of Theorem~\ref{thm:main}.

\section{Open problems}
\label{sec:conclusions}


Our research suggests interesting research directions, such as:
\begin{inparaenum}[(i)]
\item The algorithm in~\cite{DBLP:journals/jgaa/Brandenburg14} can be adjusted to compute bar $1$-visibility representations of optimal $2$-planar graphs~\cite{DBLP:journals/corr/Bekos0R17} (i.e., $2$-planar graphs with maximum density), and our construction can be also modified to obtain $1$-visible ZPRs for these graphs.  Does every $2$-planar graph admit a $1$-visible ZPR?
\item Can we generalize our result so to prove that every graph admitting a bar $1$-visibility representation also admits a $1$-visible ZPR?
\item Our algorithm computes ZPRs in which all the rectangles are intersected by the plane $Y=0$. Can this plane contain all bottom sides of the rectangles? If this is not possible, we wonder if every $1$-planar graph admits a 2.5D-visibility representation (i.e., vertices are axis-aligned boxes whose bottom faces lie on a same plane, and visibilities are both vertical and horizontal).
\end{inparaenum}


\bibliographystyle{splncs03}
\bibliography{zpr}

\clearpage
\appendix
\section*{Appendix}

\setcounter{lemma}{1}
\begin{lemma}
Both $P_{\mathcal O_1}$ and $P_{\mathcal O_2}$ are acyclic.
\end{lemma}
\begin{proof}
We  prove that  $P_{\mathcal O_1}$ is acyclic. The argument for  $P_{\mathcal O_2}$ is symmetric.
Suppose, for a contradiction, that $P_{\mathcal O_1}$ contains a directed cycle $C =\langle e_1,e_2,\dots,e_c \rangle$, as shown in Fig.~\ref{fig:cycle}.
First, note that $c>2$. If $c=2$, there are two non-homotopic parallel edges that are both part of the right path of a right wing or of a diamond in $P_\mathcal O$. But this is impossible since each pair of crossing edges in $G'$ forms an empty kite.
 Some edges of $C$ have opposite orientations in $\mathcal O$ and $\mathcal O_1$, since $\mathcal O$ is acyclic. In particular, there is at least a non-empty maximal subsequence $S = \langle e_i,e_{i+1}, \dots, e_j \rangle$ of $C$ with this property. We distinguish two cases, whether  $C$ is oriented clockwise or counter-clockwise in a closed walk along its boundary. Let $a$ and $b$ be the origin of $e_i$ and the destination of $e_j$, respectively. Note that there is a directed path from $b$ to $a$ in $P_{\mathcal O}$ (and from $a$ to $b$ in $P_{\mathcal O_1}$).

\begin{inparaenum}[{\bf Case} 1.]

\item Refer to Fig.~\ref{fig:cycle-c1}. Since $e_j$ is oriented in $P_{\mathcal O_1}$, it belongs to the right path of a right wing or of a diamond $f$ of $P_{\mathcal O}$ by {\bf R.}\ref{r:rightwing} and {\bf R.}\ref{r:diamond}. Also, $b$ is the origin of $f$, as otherwise $b$ would have an incoming edge between $e_j$ and $e_{j+1}$  in counterclockwise order from $e_j$, which violates the bimodality of the edges around $b$ or the fact that the source $s$ of $P_{\mathcal O}$ is on the outer face (recall that $e_j$ is directed outgoing from $b$ in $P_{\mathcal O}$). But then the orientation of $e_j$ in $\mathcal O_1$ contradicts {\bf R.}\ref{r:rightwing} or {\bf R.}\ref{r:diamond}.

\item Refer to Fig.~\ref{fig:cycle-c2}. Since edge $e_{i-1}$ is oriented in $P_{\mathcal O_1}$, it belongs to the right path of a right wing or of a diamond $f$ of $P_{\mathcal O}$ by {\bf R.}\ref{r:rightwing} and {\bf R.}\ref{r:diamond}. Also, $a$ is the destination of $f$, as otherwise $a$ would have an outgoing edge between $e_i$ and $e_{i-1}$  in counterclockwise order from $e_i$, which violates the bimodality of the edges around $a$ or the fact that the sink $t$ of $P_{\mathcal O}$ is on the outer face  (recall that $e_{i-1}$ is directed towards $b$ in $P_{\mathcal O}$). But then the orientation of $e_{i-1}$ in $\mathcal O_1$ contradicts {\bf R.}\ref{r:rightwing} or {\bf R.}\ref{r:diamond}.
\end{inparaenum}
\end{proof}

\begin{lemma}
$\gamma$ is a $1$-visible ZPR of $G$.
\end{lemma}
\begin{proof}
By Lemma~\ref{le:gammaprime} we know that $\gamma_2$ realizes all the edges of $G$ whose visibilities do not cross any bar in $\gamma_1$. Note that the top sides of the first and of the last vertex of $\sigma_1$ receive $y$-coordinates $n$ and $1$, respectively. Similarly,  the bottom sides of the first and of the last vertex of $\sigma_2$ receive $y$-coordinates $-n$ and $-1$, respectively. Hence all visibilities in $\gamma_2$ are preserved in $\gamma$.

Each visibility connecting a bar $b_u$ to a bar $b_v$ and traversing a bar $b_w$ in $\gamma_1$ can now be replaced with a cylinder of radius $\varepsilon < \frac{1}{2}$ and $y$-coordinate equal either to the one of the top side of $R_w$ plus $\frac{1}{2}$ or to the bottom side of $R_w$ minus $\frac{1}{2}$. In fact, the above construction ensures the top sides of $R_v$ and $R_u$ have $y$-coordinates greater than the one of $R_w$ (by at least one unit), or that the bottom sides of $R_v$ and $R_u$ have $y$-coordinates smaller than the one of $R_w$. Also, there is no rectangle $R_q$  that obstructs the visibility $(u,v)$, as otherwise $b_q$ would be traversed by $(u,v)$ in $\gamma_1$, which is not possible.

The $1$-visibility of $\gamma$ is obtained by construction, being $\gamma_1$ the intersection of $\gamma$ with the plane $Y=0$.
\end{proof}

\begin{figure}[h!]
    \centering
    \subfloat[\label{fig:s1-g}{$G$}]
    {\includegraphics[width=0.37\columnwidth,page=1]{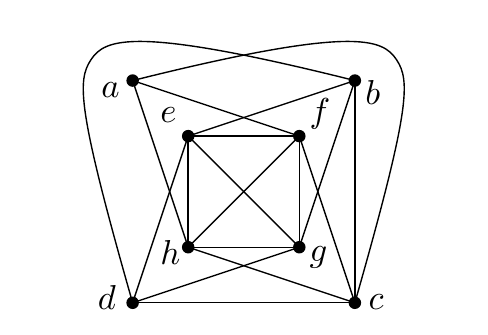}}
    \hfil
    \subfloat[\label{fig:s1-gp}{$G'$}]
    {\includegraphics[width=0.37\columnwidth,page=2]{figures/big-example}}
    \\
    \subfloat[\label{fig:s1-po}{$P_{\mathcal O}$}]
    {\includegraphics[width=0.37\columnwidth,page=3]{figures/big-example}}
    \hfil
    \subfloat[\label{fig:s1-gamma}{$\gamma_1$}]
    {\includegraphics[width=0.37\columnwidth,page=4]{figures/big-example}}
    \\
    \subfloat[\label{fig:s1-p01}{$P_{\mathcal O_1}$}]
    {\includegraphics[width=0.37\columnwidth,page=5]{figures/big-example}}
    \hfil
    \subfloat[\label{fig:s1-p02}{$P_{\mathcal O_2}$}]
    {\includegraphics[width=0.37\columnwidth,page=6]{figures/big-example}}

    \subfloat[\label{fig:s1-gammaxz}{$yz$-projection of $\gamma$}]
    {\includegraphics[width=0.37\columnwidth,page=7]{figures/big-example}}
    \caption{
    Running example for the algorithm. In (g), we only show (in red) the visibilities that cross a bar in (d). For the sake of presentation,  we chose two total orderings for $P_{\mathcal O_1}$ and $P_{\mathcal O_2}$ such that no red  visibility crosses the projection of a rectangle. The two partial orderings are $\sigma_1=\{t,a,h,f,b,e,g,s\}$ and $\sigma_2=\{h,t,g,e,b,f,a,s\}$. \label{fig:running}}
\end{figure}
\end{document}